\documentclass[12pt]{article}%
\usepackage{amsfonts}
\usepackage{mitpress}
\usepackage{amsmath}
\usepackage{amssymb}
\usepackage{graphicx}%
\providecommand{\U}[1]{\protect\rule{.1in}{.1in}}
\newtheorem{theorem}{Theorem}

\newtheorem{corollary}[theorem]{Corollary}

\newtheorem{definition}[theorem]{Definition}
\newtheorem{example}[theorem]{Example}

\newtheorem{notation}[theorem]{Notation}

\newtheorem{remark}[theorem]{Remark}

\newenvironment{proof}[1][Proof]{\noindent\textbf{#1.} }{\ \rule{0.5em}{0.5em}}
\newdimen\dummy
\dummy=\oddsidemargin
\begin{document}

\title{On the basins of attraction of the regular autonomous asynchronous systems}
\author{Serban E. Vlad\\Str. Zimbrului, Nr. 3, Bl. PB68, Ap. 11, 410430, Oradea, Romania, mail:
serban\_e\_vlad@yahoo.com, web page: www.serbanvlad.ro}
\maketitle

\begin{abstract}
The Boolean autonomous dynamical systems, also called regular autonomous
asynchronous systems are systems whose `vector field' is a function
$\Phi:\{0,1\}^{n}\rightarrow\{0,1\}^{n}$ and time is discrete or continuous.
While the synchronous systems have their coordinate functions $\Phi
_{1},...,\Phi_{n}$ computed at the same time: $\Phi,\Phi\circ\Phi,\Phi
\circ\Phi\circ\Phi,...$ the asynchronous systems have $\Phi_{1},...,\Phi_{n}$
computed independently on each other. The purpose of the paper is that of
studying the basins of attraction of the fixed points, of the orbits and of
the $\omega$-limit sets of the regular autonomous asynchronous systems, by
continuing the study started in \cite{bib8}. The bibliography consists in analogies.

\end{abstract}


\textbf{keywords}: asynchronous system, $\omega-$limit set, invariance, basin
of attraction

\section{Introduction}

The $\mathbf{R}\rightarrow\{0,1\}$ functions model the digital electrical
signals and they are not studied in literature. An asynchronous circuit
without input, considered as a collection of $n$ signals, should be
deterministically modelled by a function $x:\mathbf{R}\rightarrow\{0,1\}^{n}$
called state. Several parameters related with the asynchronous circuit are
either unknown, or perhaps variable or simply ignored in modeling: the
temperature, the tension of the mains, the delays the occur in the computation
of the Boolean functions etc. For this reason, instead of a function $x$ we
have in general a set $X$ of functions $x,$ called state space or autonomous
system, where each $x$ represents a possibility of modeling the circuit. When
$X$ is constructed by making use of a 'vector field' $\Phi:\{0,1\}^{n}%
\rightarrow\{0,1\}^{n},$ the system $X$ is called regular. The universal
regular autonomous asynchronous systems are the Boolean dynamical systems and
they are identified with $\Phi$.

The dynamics of these systems is described by the so called state portraits.
We give the example of the function $\Phi:\{0,1\}^{2}\rightarrow\{0,1\}^{2}$
that is defined by Table 1, where $\mu=(\mu_{1},\mu_{2})\in\{0,1\}^{2}:$%
\begin{align*}
&
\begin{array}
[c]{cc}%
(\mu_{1},\mu_{2}) & (\Phi_{1}(\mu_{1},\mu_{2}),\Phi_{2}(\mu_{1},\mu_{2}))\\
(0,0) & (1,1)\\
(0,1) & (1,1)\\
(1,0) & (1,0)\\
(1,1) & (0,1)
\end{array}
\\
&  \;\;\;\;\;\;\;\;\;\;\;\;\;Table\;1
\end{align*}
The state portrait of $\Phi$ was drawn in Figure \ref{attractive3}
\begin{figure}
[ptb]
\begin{center}
\fbox{\includegraphics[
height=0.9003in,
width=1.3534in
]%
{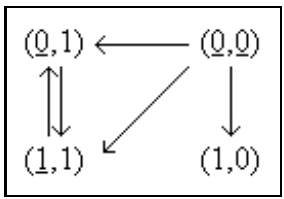}%
}\caption{Example of state portrait}%
\label{attractive3}%
\end{center}
\end{figure}
where the arrows show the increase of time. The coordinates $\mu_{i}%
,i\in\{1,2\}$ are underlined if $\Phi_{i}(\mu_{1},\mu_{2})\neq\mu_{i}$ and
they are called unstable, or enabled, or excited. These are the coordinates
that are about to change their value. The coordinates $\mu_{i}$ that are not
underlined satisfy by definition $\Phi_{i}(\mu_{1},\mu_{2})=\mu_{i}$ and they
are called stable, or disabled, or not excited. These are the coordinates that
cannot change their value. Three arrows start from the point $(0,0)$ where
both coordinates are unstable, showing the fact that $\Phi_{1}(0,0)$ may be
computed first, $\Phi_{2}(0,0)$ may be computed first or $\Phi_{1}%
(0,0),\Phi_{2}(0,0)$ may be computed simultaneously. Note that the system was
identified with the function $\Phi$.

The existence of several possibilities of evolution of the system (three
possibilities in $(0,0)$) is the key characteristic of asynchronicity, as
opposed to synchronicity where the coordinates $\Phi_{i}(\mu)$ are always
computed simultaneously, $i\in\{1,...,n\}$ for all $\mu\in\{0,1\}^{n}$ and the
system's run is: $\mu,\Phi(\mu),(\Phi\circ\Phi)(\mu),...,(\Phi\circ
...\circ\Phi)(\mu),...$

The purpose of the paper is that of defining in the asynchronous systems
theory, by following analogies, the basins of attraction of the fixed points
and of the orbits from the dynamical systems theory. We shall also define the
basins of attraction of the $\omega-$limit sets. The paper continues the study
of the basins of attraction that was started in \cite{bib8} and many
introductory issues are taken from that work.

\section{Preliminaries}

\begin{notation}
The set $\mathbf{B}=\{0,1\}$ is the binary Boole algebra, endowed with the
usual algebraical laws%
\begin{align*}
&
\begin{array}
[c]{cc}
& \overline{\ \ }\\
0 & 1\\
1 & 0
\end{array}
,\;\;\;\;\;\;%
\begin{array}
[c]{ccc}%
\cdot & 0 & 1\\
0 & 0 & 0\\
1 & 0 & 1
\end{array}
,\;\;\;\;\;\;%
\begin{array}
[c]{ccc}%
\cup & 0 & 1\\
0 & 0 & 1\\
1 & 1 & 1
\end{array}
,\;\;\;\;\;\;%
\begin{array}
[c]{ccc}%
\oplus & 0 & 1\\
0 & 0 & 1\\
1 & 1 & 0
\end{array}
\\
&  \;\;\;\;\;\;\;\;\;\;\;\;\;\;\;\;\;\;\ \ \ \ \ Table\;2
\end{align*}
and with the discrete topology.
\end{notation}

\begin{definition}
The sequence $\alpha:\mathbf{N}\rightarrow\mathbf{B}^{n},$ usually denoted by
$\alpha^{k},k\in\mathbf{N},$ is called \textbf{progressive} if the sets%
\[
\{k|k\in\mathbf{N},\alpha_{i}^{k}=1\}
\]
are infinite for all $i\in\{1,...,n\}.$ We denote the set of the progressive
sequences by $\Pi_{n}.$
\end{definition}

\begin{definition}
For the function $\Phi:\mathbf{B}^{n}\rightarrow\mathbf{B}^{n}$ and $\nu
\in\mathbf{B}^{n}$ we define $\Phi^{\nu}:\mathbf{B}^{n}\rightarrow
\mathbf{B}^{n}$ by $\forall\mu\in\mathbf{B}^{n},$%
\[
\Phi^{\nu}(\mu)=(\overline{\nu_{1}}\cdot\mu_{1}\oplus\nu_{1}\cdot\Phi_{1}%
(\mu),...,\overline{\nu_{n}}\cdot\mu_{n}\oplus\nu_{n}\cdot\Phi_{n}(\mu)).
\]

\end{definition}

\begin{definition}
The functions $\Phi^{\alpha^{0}...\alpha^{k}}:\mathbf{B}^{n}\rightarrow
\mathbf{B}^{n}$ are defined for $k\in\mathbf{N}$ and $\alpha^{0}%
,...,\alpha^{k}\in\mathbf{B}^{n}$ iteratively: $\forall\mu\in\mathbf{B}^{n},$%
\[
\Phi^{\alpha^{0}...\alpha^{k}\alpha^{k+1}}(\mu)=\Phi^{\alpha^{k+1}}%
(\Phi^{\alpha^{0}...\alpha^{k}}(\mu)).
\]

\end{definition}

\begin{notation}
We denote by $\chi_{A}:\mathbf{R}\rightarrow\mathbf{B}$ the characteristic
function of the set $A\subset\mathbf{R}$: $\forall t\in\mathbf{R},$%
\[
\chi_{A}(t)=\left\{
\begin{array}
[c]{c}%
1,t\in A\\
0,t\notin A
\end{array}
\right.  .
\]

\end{notation}

\begin{notation}
We denote by $Seq$ the set of the sequences $t_{0}<t_{1}<...<t_{k}<...$ of
real numbers that are unbounded from above.
\end{notation}

\begin{definition}
The functions $\rho:\mathbf{R}\rightarrow\mathbf{B}^{n}$ of the form $\forall
t\in\mathbf{R},$%
\begin{equation}
\rho(t)=\alpha^{0}\cdot\chi_{\{t_{0}\}}(t)\oplus\alpha^{1}\cdot\chi
_{\{t_{1}\}}(t)\oplus...\oplus\alpha^{k}\cdot\chi_{\{t_{k}\}}(t)\oplus...
\label{1}%
\end{equation}
where $\alpha\in\Pi_{n}$ and $(t_{k})\in Seq$ are called \textbf{progressive}
and their set is denoted by $P_{n}.$
\end{definition}

\begin{definition}
The function $\Phi^{\rho}:\mathbf{B}^{n}\times\mathbf{R}\rightarrow
\mathbf{B}^{n}$ that is defined in the following way%
\[
\Phi^{\rho}(\mu,t)=\mu\cdot\chi_{(-\infty,t_{0})}(t)\oplus\Phi^{\alpha^{0}%
}(\mu)\cdot\chi_{\lbrack t_{0},t_{1})}(t)\oplus\Phi^{\alpha^{0}\alpha^{1}}%
(\mu)\cdot\chi_{\lbrack t_{1},t_{2})}(t)\oplus...
\]%
\[
...\oplus\Phi^{\alpha^{0}...\alpha^{k}}(\mu)\cdot\chi_{\lbrack t_{k},t_{k+1}%
)}(t)\oplus...
\]
is called \textbf{flow}, \textbf{motion} or \textbf{orbit} (of $\mu
\in\mathbf{B}^{n}$). We have supposed that $\rho\in P_{n}$ is like at (\ref{1}).
\end{definition}

\begin{definition}
The set%
\[
Or_{\rho}(\mu)=\{\Phi^{\rho}(\mu,t)|t\in\mathbf{R}\}
\]
is also called \textbf{orbit} (of $\mu\in\mathbf{B}^{n}$).
\end{definition}

\begin{remark}
The function $\Phi^{\nu}$ shows how an asynchronous iteration of $\Phi$ is
made: for any $i\in\{1,...,n\},$ if $\nu_{i}=0$ then $\Phi_{i}$ is not
computed, since $\Phi_{i}^{\nu}(\mu)=\mu_{i}$ and if $\nu_{i}=1$ then
$\Phi_{i}$ is computed, since $\Phi_{i}^{\nu}(\mu)=\Phi_{i}(\mu).$

The definition of $\Phi^{\alpha^{0}...\alpha^{k}}$ generalizes this idea to an
arbitrary number $k+1$ of asynchronous iterations, with the supplementary
request that each coordinate $\Phi_{i}$ is computed infinitely many times in
the sequence $\mu,\Phi^{\alpha^{0}}(\mu),\Phi^{\alpha^{0}\alpha^{1}}%
(\mu),...,\Phi^{\alpha^{0}...\alpha^{k}}(\mu),...$ whenever $\alpha\in\Pi
_{n}.$

The sequences $(t_{k})\in Seq$ make the pass from the discrete time
$\mathbf{N}$ to the continuous time $\mathbf{R}$ and each $\rho\in P_{n}$
shows, in addition to $\alpha\in\Pi_{n}$, the time instants $t_{k}$ when
$\Phi$ is computed (asynchronously). Thus $\Phi^{\rho}(\mu,t),t\in\mathbf{R}$
is the continuous time computation of the sequence $\mu,$ $\Phi^{\alpha^{0}%
}(\mu),$ $\Phi^{\alpha^{0}\alpha^{1}}(\mu),$ $...,$ $\Phi^{\alpha^{0}%
...\alpha^{k}}(\mu),$ $...$

When $\alpha$ runs in $\Pi_{n}$ and $(t_{k})$ runs in $Seq$ we get the
'unbounded delay model' of computation of the Boolean function $\Phi$,
represented in discrete time by the sequences $\mu,\Phi^{\alpha^{0}}(\mu
),\Phi^{\alpha^{0}\alpha^{1}}(\mu),...,\Phi^{\alpha^{0}...\alpha^{k}}%
(\mu),...$ and in continuous time by the orbits $\Phi^{\rho}(\mu,t)$
respectively. We shall not insist on the non-formalized way that the engineers
describe this model; we just mention that the 'unbounded delay model' is a
reasonable way of starting the analysis of a circuit in which the delays
occurring in the computation of the Boolean functions $\Phi$ are arbitrary
positive numbers. If we restrict suitably the ranges of $\alpha$ and $(t_{k})$
we get the 'bounded delay model' of computation of $\Phi$ and if both $\alpha
$, $(t_{k})$ are fixed, then we obtain the 'fixed delay model' of computation
of $\Phi,$ determinism.
\end{remark}

\begin{theorem}
\label{The10}\cite{bib8} Let $\alpha\in\Pi_{n},(t_{k})\in Seq$ be arbitrary
and the function%
\[
\rho(t)=\alpha^{0}\cdot\chi_{\{t_{0}\}}(t)\oplus\alpha^{1}\cdot\chi
_{\{t_{1}\}}(t)\oplus...\oplus\alpha^{k}\cdot\chi_{\{t_{k}\}}(t)\oplus...,
\]
$\rho\in P_{n}.$ The following statements are true:

a) $\{\alpha^{k}|k\geq k_{1}\}\in\Pi_{n}$ for any $k_{1}\in\mathbf{N};$

b) $(t_{k})\cap(t^{\prime},\infty)\in Seq$ for any $t^{\prime}\in\mathbf{R};$

c) $\rho\cdot\chi_{(t^{\prime},\infty)}\in P_{n}$ for any $t^{\prime}%
\in\mathbf{R};$

d) $\forall\mu\in\mathbf{B}^{n},\forall\mu^{\prime}\in\mathbf{B}^{n},\forall
t^{\prime}\in\mathbf{R},$%
\[
\Phi^{\rho}(\mu,t^{\prime})=\mu^{\prime}\Longrightarrow\forall t\geq
t^{\prime},\Phi^{\rho}(\mu,t)=\Phi^{\rho\cdot\chi_{(t^{\prime},\infty)}}%
(\mu^{\prime},t).
\]

\end{theorem}

\begin{notation}
For any $d\in\mathbf{R,}$ we denote with $\tau^{d}:\mathbf{R}\rightarrow
\mathbf{R}$ the translation $\forall t\in\mathbf{R},\tau^{d}(t)=t-d.$
\end{notation}

\begin{theorem}
\label{The11}\cite{bib8} Let be $\mu\in\mathbf{B}^{n},\rho\in P_{n}$ and
$d\in\mathbf{R}.$ The function $\rho\circ\tau^{d}$ is progressive and we have%
\[
\Phi^{\rho\circ\tau^{d}}(\mu,t)=\Phi^{\rho}(\mu,t-d).
\]

\end{theorem}

\begin{definition}
The \textbf{universal regular autonomous asynchronous system} that is
generated by $\Phi:\mathbf{B}^{n}\rightarrow\mathbf{B}^{n}$ is by definition%
\[
\Xi_{\Phi}=\{\Phi^{\rho}(\mu,\cdot)|\mu\in\mathbf{B}^{n},\rho\in P_{n}\};
\]
any $x(t)=\Phi^{\rho}(\mu,t)$ is called \textbf{state} (of $\Xi_{\Phi}$),
$\mu$ is called \textbf{initial value} (of $x$), or \textbf{initial state} (of
$\Xi_{\Phi}$) and $\Phi$ is called \textbf{generator function} (of $\Xi_{\Phi
}$).
\end{definition}

\begin{remark}
The asynchronous systems are non-deterministic in general, due to the
uncertainties that occur in the modeling of the asynchronous circuits.
\ Non-determinism is produced, in the case of $\Xi_{\Phi}$, by the fact that
the initial state $\mu$ and the way $\rho$ of iterating $\Phi$ are not known.

Some notes on the terminology:

- universality means the greatest in the sense of inclusion. Any $X\subset
\Xi_{\Phi}$ is a system, but we shall not study such systems in the present paper;

- regularity means the existence of a generator function $\Phi$, i.e.
analogies with the dynamical systems theory;

- autonomy means here that no input exists. We mention the fact that autonomy
has another non-equivalent definition also, a system is called autonomous if
its input set has exactly one element;

- asynchronicity refers (vaguely) to the fact that the coordinate functions
$\Phi_{1},...,\Phi_{n}$ are computed independently on each other. Its antonym
synchronicity means that the iterates of $\Phi$ are: $\Phi,\Phi\circ
\Phi,...,\Phi\circ...\circ\Phi,...$ i.e. in the sequence $\Phi^{\alpha^{0}%
},\Phi^{\alpha^{0}\alpha^{1}},...,\Phi^{\alpha^{0}...\alpha^{k}},...$ all
$\alpha^{k}$ are $(1,...,1),k\in\mathbf{N}.$ That is the discrete time of the
dynamical systems.
\end{remark}

\begin{definition}
Let $x:\mathbf{R}\rightarrow\mathbf{B}^{n}$ be some function. If
\[
\exists t^{\prime}\in\mathbf{R},\forall t\geq t^{\prime},x(t)=x(t^{\prime}),
\]
we say that \textbf{the limit} $\underset{t\rightarrow\infty}{\lim}x(t)$ (or
the \textbf{final value} of $x$) \textbf{exists} and we denote%
\[
\underset{t\rightarrow\infty}{\lim}x(t)=x(t^{\prime}).
\]

\end{definition}

\begin{theorem}
\label{The15}\cite{bib7},\cite{bib8} $\forall\mu\in\mathbf{B}^{n},\forall
\mu^{\prime}\in\mathbf{B}^{n},\forall\rho\in P_{n},$%
\[
\underset{t\rightarrow\infty}{\lim}\Phi^{\rho}(\mu,t)=\mu^{\prime
}\Longrightarrow\Phi(\mu^{\prime})=\mu^{\prime},
\]
i.e. if the final value of $\Phi^{\rho}(\mu,\cdot)$ exists, it is a fixed
point of $\Phi$.
\end{theorem}

\begin{theorem}
\label{The16}\cite{bib7},\cite{bib8} $\forall\mu\in\mathbf{B}^{n},\forall
\mu^{\prime}\in\mathbf{B}^{n},\forall\rho\in P_{n},$%
\[
(\Phi(\mu^{\prime})=\mu^{\prime}\;and\;\exists t^{\prime}\in\mathbf{R}%
,\Phi^{\rho}(\mu,t^{\prime})=\mu^{\prime})\Longrightarrow\forall t\geq
t^{\prime},\Phi^{\rho}(\mu,t)=\mu^{\prime},
\]
meaning that if the fixed point $\mu^{\prime}$ of $\Phi$ is accessible, then
it is the final value of $\Phi^{\rho}(\mu,\cdot).$
\end{theorem}

\begin{corollary}
\label{Cor17}\cite{bib8} We have $\forall\mu\in\mathbf{B}^{n},\forall\rho\in
P_{n},$%
\[
\Phi(\mu)=\mu\Longrightarrow\forall t\in\mathbf{R},\Phi^{\rho}(\mu,t)=\mu.
\]

\end{corollary}

\section{$\omega-$limit sets}

\begin{definition}
\label{Def19}For $\mu\in\mathbf{B}^{n}$ and $\rho\in P_{n},$ the set%
\[
\omega_{\rho}(\mu)=\{\mu^{\prime}|\mu^{\prime}\in\mathbf{B}^{n},\exists
(t_{k})\in Seq,\underset{k\rightarrow\infty}{\lim}\Phi^{\rho}(\mu,t_{k}%
)=\mu^{\prime}\}
\]
is called the $\omega-$\textbf{limit set} of the orbit $\Phi^{\rho}(\mu
,\cdot).$
\end{definition}

\begin{theorem}
\label{The20_}\cite{bib8} For any $\mu\in\mathbf{B}^{n}$ and any $\rho\in
P_{n},$ we have:

a) $\omega_{\rho}(\mu)\neq\emptyset;$

b) $\forall t^{\prime}\in\mathbf{R},$ $\omega_{\rho}(\mu)\subset\{\Phi^{\rho
}(\mu,t)|t\geq t^{\prime}\}\subset Or_{\rho}(\mu);$

c) $\exists t^{\prime}\in\mathbf{R},\omega_{\rho}(\mu)=\{\Phi^{\rho}%
(\mu,t)|t\geq t^{\prime}\}$ and any $t^{\prime\prime}\geq t^{\prime}$ fulfills
$\omega_{\rho}(\mu)=\{\Phi^{\rho}(\mu,t)|t\geq t^{\prime\prime}\};$

d) $\forall t^{\prime}\in\mathbf{R},\forall t^{\prime\prime}\geq t^{\prime
},\{\Phi^{\rho}(\mu,t)|t\geq t^{\prime}\}=\{\Phi^{\rho}(\mu,t)|t\geq
t^{\prime\prime}\}$ implies $\omega_{\rho}(\mu)=\{\Phi^{\rho}(\mu,t)|t\geq
t^{\prime}\};$

e) we suppose that $\omega_{\rho}(\mu)=\{\Phi^{\rho}(\mu,t)|t\geq t^{\prime
}\},t^{\prime}\in\mathbf{R}.$ Then $\forall\mu^{\prime}\in\omega_{\rho}%
(\mu),\forall t^{\prime\prime}\geq t^{\prime},$ if $\Phi^{\rho}(\mu
,t^{\prime\prime})=\mu^{\prime}$ we get%
\[
\omega_{\rho}(\mu)=\{\Phi^{\rho\cdot\chi_{(t^{\prime\prime},\infty)}}%
(\mu^{\prime},t)|t\geq t^{\prime\prime}\}=Or_{\rho\cdot\chi_{(t^{\prime\prime
},\infty)}}(\mu^{\prime})=\omega_{\rho\cdot\chi_{(t^{\prime\prime},\infty)}%
}(\mu^{\prime}).
\]

\end{theorem}

\begin{remark}
\label{Rem23}If in Theorem \ref{The20_} e) we take $t^{\prime\prime}%
\in\mathbf{R}$ arbitrarily, the equation%
\begin{equation}
\omega_{\rho}(\mu)=\omega_{\rho\cdot\chi_{(t^{\prime\prime},\infty)}}%
(\Phi^{\rho}(\mu,t^{\prime\prime})) \label{asterisc}%
\end{equation}
is still true. Indeed, for sufficiently great $t^{\prime\prime\prime}$, the
terms in (\ref{asterisc}) are equal with%
\[
\{\Phi^{\rho}(\mu,t)|t\geq t^{\prime\prime\prime}\}=\{\Phi^{\rho\cdot
\chi_{(t^{\prime\prime},\infty)}}(\Phi^{\rho}(\mu,t^{\prime\prime}),t)|t\geq
t^{\prime\prime\prime}\}.
\]

\end{remark}

\begin{theorem}
\label{The20}\cite{bib8} For arbitrary $\mu\in\mathbf{B}^{n}$,$\rho\in P_{n}$
the following statements are true:

a) $\underset{t\rightarrow\infty}{\lim}\Phi^{\rho}(\mu,t)$ exists
$\Longleftrightarrow card(\omega_{\rho}(\mu))=1;$

b) if $\exists\mu^{\prime}\in\mathbf{B}^{n},\omega_{\rho}(\mu)=\{\mu^{\prime
}\},$ then $\underset{t\rightarrow\infty}{\lim}\Phi^{\rho}(\mu,t)=\mu^{\prime
}$ and $\Phi(\mu^{\prime})=\mu^{\prime};$

c) if $\exists\mu^{\prime}\in\mathbf{B}^{n},\Phi(\mu^{\prime})=\mu^{\prime}$
and $\mu^{\prime}\in Or_{\rho}(\mu),$ then $\omega_{\rho}(\mu)=\{\mu^{\prime
}\}.$
\end{theorem}

\begin{theorem}
\label{The22}\cite{bib8} Let be $\mu\in\mathbf{B}^{n},\rho\in P_{n}%
,d\in\mathbf{R}.$ We have $\omega_{\rho}(\mu)=\omega_{\rho\circ\tau^{d}}%
(\mu).$
\end{theorem}

\section{Invariant sets}

\begin{theorem}
\label{The14}\cite{bib8} We consider the function $\Phi:\mathbf{B}%
^{n}\rightarrow\mathbf{B}^{n}$ and let be the set $A\in P^{\ast}%
(\mathbf{B}^{n}).$ For any $\mu\in A,$ the following properties are equivalent%
\begin{equation}
\exists\alpha\in\Pi_{n},\forall k\in\mathbf{N},\Phi^{\alpha^{0}...\alpha^{k}%
}(\mu)\in A, \label{inv3__}%
\end{equation}%
\begin{equation}
\exists\rho\in P_{n},\forall t\in\mathbf{R},\Phi^{\rho}(\mu,t)\in A,
\label{inv1__}%
\end{equation}%
\begin{equation}
\exists\rho\in P_{n},Or_{\rho}(\mu)\subset A \label{inv2}%
\end{equation}
and the following properties are also equivalent%
\begin{equation}
\forall\alpha\in\Pi_{n},\forall k\in\mathbf{N},\Phi^{\alpha^{0}...\alpha^{k}%
}(\mu)\in A, \label{bas5}%
\end{equation}%
\begin{equation}
\forall\rho\in P_{n},\forall t\in\mathbf{R},\Phi^{\rho}(\mu,t)\in A,
\label{inv2__}%
\end{equation}%
\begin{equation}
\forall\rho\in P_{n},Or_{\rho}(\mu)\subset A, \label{inv2_}%
\end{equation}%
\begin{equation}
\forall\lambda\in\mathbf{B}^{n},\Phi^{\lambda}(\mu)\in A. \label{inv_1_}%
\end{equation}

\end{theorem}

\begin{definition}
\label{Def166}The set $A\in P^{\ast}(\mathbf{B}^{n})$ is called a
\textbf{p-invariant} (or \textbf{p-stable}) \textbf{set} of the system
$\Xi_{\Phi}$ if it fulfills for any $\mu\in A$ one of (\ref{inv3__}),...,
(\ref{inv2}) and it is called an \textbf{n-invariant }(or n-\textbf{stable})
\textbf{set} of $\Xi_{\Phi}$ if it fulfills $\forall\mu\in A$ one of
(\ref{bas5}),..., (\ref{inv_1_}).
\end{definition}

\begin{remark}
In the previous terminology, the letter 'p' comes from 'possibly' and the
letter 'n' comes from 'necessarily'. Both 'p' and 'n' refer to the
quantification of $\rho$. Such kind of p-definitions and n-definitions
recalling logic are caused by the fact that we translate 'real' concepts into
'binary' concepts and the former have no $\rho$ parameters, thus after
translation $\rho$ may appear quantified in two ways. The obvious implication
is n-invariance $\Longrightarrow$ p-invariance.
\end{remark}

\begin{theorem}
\label{The81}\cite{bib8} Let be $\mu\in\mathbf{B}^{n}$ and $\rho^{\prime}\in
P_{n}.$

a) If $\Phi(\mu)=\mu,$ then $\{\mu\}$ is an n-invariant set and the set $Eq$
of the fixed points of $\Phi$ is also n-invariant;

b) the set $Or_{\rho^{\prime}}(\mu)$ is p-invariant and $\underset{\rho\in
P_{n}}{\bigcup}Or_{\rho}(\mu)$ is n-invariant;

c) the set $\omega_{\rho^{\prime}}(\mu)$ is p-invariant.
\end{theorem}

\section{The basin of attraction}

\begin{theorem}
\label{The138}\cite{bib8} We consider the set $A\in P^{\ast}(\mathbf{B}^{n}).$
For any $\mu\in\mathbf{B}^{n},$ the following statements are equivalent%
\begin{equation}
\exists\alpha\in\Pi_{n},\exists k^{\prime}\in\mathbf{N},\forall k\geq
k^{\prime},\Phi^{\alpha^{0}...\alpha^{k}}(\mu)\in A,
\end{equation}%
\begin{equation}
\exists\rho\in P_{n},\exists t^{\prime}\in R,\forall t\geq t^{\prime}%
,\Phi^{\rho}(\mu,t)\in A, \label{sta3_}%
\end{equation}%
\begin{equation}
\exists\rho\in P_{n},\omega_{\rho}(\mu)\subset A \label{sta9}%
\end{equation}
and the following statements are equivalent too%
\begin{equation}
\forall\alpha\in\Pi_{n},\exists k^{\prime}\in\mathbf{N},\forall k\geq
k^{\prime},\Phi^{\alpha^{0}...\alpha^{k}}(\mu)\in A,
\end{equation}%
\begin{equation}
\forall\rho\in P_{n},\exists t^{\prime}\in R,\forall t\geq t^{\prime}%
,\Phi^{\rho}(\mu,t)\in A, \label{sta13}%
\end{equation}%
\begin{equation}
\forall\rho\in P_{n},\omega_{\rho}(\mu)\subset A. \label{sta14}%
\end{equation}

\end{theorem}

\begin{definition}
\label{Def123}The \textbf{basin }(or \textbf{kingdom}, or \textbf{domain})
\textbf{of p-attraction} or the \textbf{p-stable set} of the set $A\in
P^{\ast}(\mathbf{B}^{n})$ is given by%
\begin{equation}
\overline{W}(A)=\{\mu|\mu\in\mathbf{B}^{n},\exists\rho\in P_{n},\omega_{\rho
}(\mu)\subset A\};
\end{equation}
the \textbf{basin }(or \textbf{kingdom}, or \textbf{domain}) \textbf{of
n-attraction} or the \textbf{n-stable set} of the set $A$ is given by%
\begin{equation}
\underline{W}(A)=\{\mu|\mu\in\mathbf{B}^{n},\forall\rho\in P_{n},\omega_{\rho
}(\mu)\subset A\}.
\end{equation}

\end{definition}

\begin{remark}
Definition \ref{Def123} makes use of the properties (\ref{sta9}) and
(\ref{sta14}). We can make use also in this Definition of the other equivalent
properties from Theorem \ref{The138}.

In Definition \ref{Def123}, one or both basins of attraction $\overline
{W}(A),\underline{W}(A)$ may be empty.
\end{remark}

\begin{theorem}
\label{The25}\cite{bib8} We have:

i) $\overline{W}(\mathbf{B}^{n})=\underline{W}(\mathbf{B}^{n})=\mathbf{B}%
^{n};$

ii) if $A\subset A^{\prime},$ then $\overline{W}(A)\subset\overline
{W}(A^{\prime})$ and $\underline{W}(A)\subset\underline{W}(A^{\prime})$ hold.
\end{theorem}

\begin{definition}
\label{Def168_}When $\overline{W}(A)\neq\emptyset,$ $A$ is said to be
\textbf{p-attractive} and for any non-empty set $B\subset\overline{W}(A),$ we
say that $A$ is \textbf{p-attractive} for $B$ and that $B$ is
\textbf{p-attracted} by $A$; $A$ is by definition \textbf{partially
p-attractive} if $\overline{W}(A)\notin\{\emptyset,\mathbf{B}^{n}\}$ and
\textbf{totally p-attractive} whenever $\overline{W}(A)=\mathbf{B}^{n}.$

The fact that $\underline{W}(A)\neq\emptyset$ makes us say that $A$ is
\textbf{n-attractive} and in this situation for any non-empty $B\subset
\underline{W}(A),$ $A$ is called \textbf{n-attractive} for $B$ and $B$ is
called to be \textbf{n-attracted} by $A;$ we use to say that $A$ is
\textbf{partially n-attractive} if $\underline{W}(A)\notin\{\emptyset
,\mathbf{B}^{n}\}$ and \textbf{totally n-attractive} if $\underline
{W}(A)=\mathbf{B}^{n}.$
\end{definition}

\begin{theorem}
\label{The106}\cite{bib8} Let $A\in P^{\ast}(\mathbf{B}^{n})$ be some set. If
$A$ is p-invariant, then $A\subset\overline{W}(A)$ and $A$ is also
p-attractive; if $A$ is n-invariant, then $A\subset\underline{W}(A)$ and $A$
is also n-attractive.
\end{theorem}

\begin{remark}
The previous Theorem shows the connection that exists between invariance and
attractiveness. If $A$ is p-attractive, then $\overline{W}(A)$ is the greatest
set that is p-attracted by $A$ and the point is that this really happens when
$A$ is p-invariant. The other situation is dual.
\end{remark}

\begin{theorem}
\label{The140}\cite{bib8} Let be $A\in P^{\ast}(\mathbf{B}^{n}).$ If $A$ is
p-attractive, then $\overline{W}(A)$ is p-invariant and if $A$ is
n-attractive, then $\underline{W}(A)$ is n-invariant.
\end{theorem}

\begin{corollary}
\cite{bib8} If the set $A\in P^{\ast}(\mathbf{B}^{n})$ is p-invariant, then
$\overline{W}(A)$ is p-invariant and if $A$ is n-invariant, then the basin of
n-attraction $\underline{W}(A)$ is n-invariant.
\end{corollary}

\section{The basin of attraction of the fixed points}

\begin{notation}
For any point $\mu\in\mathbf{B}^{n}$ we use the simpler notations
$\overline{W}(\mu),$ $\underline{W}(\mu)$ instead of $\overline{W}(\{\mu\}),$
$\underline{W}(\{\mu\}).$ Furthermore, if the point $\mu$ is identified with
the $n-$tuple $(\mu_{1},...,\mu_{n}),$ it is usual to write $\overline{W}%
(\mu_{1},...,\mu_{n}),$ $\underline{W}(\mu_{1},...,\mu_{n})$ for these sets.
\end{notation}

\begin{remark}
This section is dedicated to the special case when in Definition \ref{Def123}
the set $A\in P^{\ast}(\mathbf{B}^{n})$ consists in a point $\mu,$ in other
words%
\begin{equation}
\overline{W}(\mu)=\{\mu^{\prime}|\mu^{\prime}\in\mathbf{B}^{n},\exists
\rho^{\prime}\in P_{n},\omega_{\rho^{\prime}}(\mu^{\prime})\subset\{\mu\}\},
\label{bf1}%
\end{equation}%
\begin{equation}
\underline{W}(\mu)=\{\mu^{\prime}|\mu^{\prime}\in\mathbf{B}^{n},\forall
\rho^{\prime}\in P_{n},\omega_{\rho^{\prime}}(\mu^{\prime})\subset\{\mu\}\}.
\label{bf2}%
\end{equation}
The fact that the point $\mu$ is chosen to be fixed is justified by the
\end{remark}

\begin{theorem}
$\overline{W}(\mu)\neq\emptyset\Longleftrightarrow\mu$ is a fixed point of
$\Phi$ and similarly, $\underline{W}(\mu)\neq\emptyset\Longleftrightarrow\mu$
is a fixed point of $\Phi.$
\end{theorem}

\begin{proof}
We prove the first statement. If $\mu^{\prime}\in\overline{W}(\mu),$ then
$\rho^{\prime}\in P_{n}$ exists such that $\omega_{\rho^{\prime}}(\mu^{\prime
})\subset\{\mu\}.$ In this case $\omega_{\rho^{\prime}}(\mu^{\prime})$ is
non-empty, thus $\omega_{\rho^{\prime}}(\mu^{\prime})=\{\mu\}$ and, from
Theorem \ref{The20} b), $\Phi(\mu)=\mu.$

Let us suppose now that $\Phi(\mu)=\mu.$ For any $\rho^{\prime}\in P_{n},$
$Or_{\rho^{\prime}}(\mu)=\omega_{\rho^{\prime}}(\mu)=\{\mu\}$ from Corollary
\ref{Cor17}, thus $\mu\in\overline{W}(\mu)$ and $\overline{W}(\mu
)\neq\emptyset.$
\end{proof}

\begin{remark}
In \cite{bib5} at page 5, the fixed point $x_{0}\in X$ is called attractive if
the neighborhood $U\subset X$ and $t^{\prime}>0$ exist such that%
\[
\forall x\in U,\forall t>t^{\prime},\Phi_{t}(x)\in U\text{ and }%
\underset{t\rightarrow\infty}{\lim}|\Phi_{t}(x)-x_{0}|=0.
\]

We give also the point of view from \cite{bib6} where, at page 110 it is said,
in a discrete time context, that the basin of attraction of an attractive
fixed point $x_{0}\in X$ is formed by the the set of all the initial points of
some sequences of iterates that converge to $x_{0}$.
\end{remark}

\begin{example}
In Figure \ref{fixed1_} we have the property that all the points are fixed
points and $\forall\mu\in\mathbf{B}^{2},\forall\rho\in P_{2},$
\begin{figure}
[ptb]
\begin{center}
\fbox{\includegraphics[
height=0.8691in,
width=1.3033in
]%
{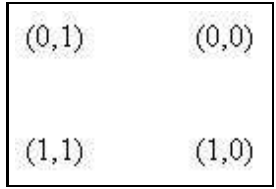}%
}\caption{The basins of attraction of the fixed points}%
\label{fixed1_}%
\end{center}
\end{figure}
\[
\overline{W}(\mu)=\underline{W}(\mu)=\{\mu\}.
\]
Any $\mu\in\mathbf{B}^{2}$ is partially p-attractive and partially n-attractive.
\end{example}

\begin{example}
The point $(1,0)$ is fixed in Figure \ref{fixed2} and
\begin{figure}
[ptb]
\begin{center}
\fbox{\includegraphics[
height=0.8821in,
width=1.241in
]%
{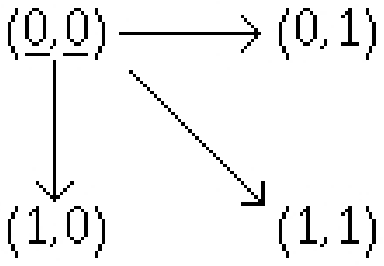}%
}\caption{The basins of attraction of the fixed points}%
\label{fixed2}%
\end{center}
\end{figure}
\[
\overline{W}(1,0)=\{(0,0),(1,0)\},
\]%
\[
\underline{W}(1,0)=\{(1,0)\}.
\]
The point $(1,0)$ is partially p-attractive and partially n-attractive.
\end{example}

\begin{example}
We have also the example when in Figure \ref{fixed3}:
\begin{figure}
[ptb]
\begin{center}
\fbox{\includegraphics[
height=0.8951in,
width=1.2324in
]%
{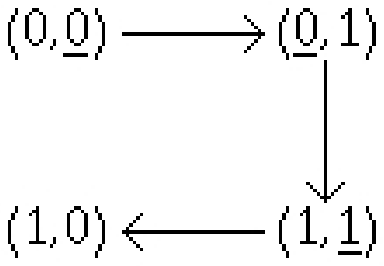}%
}\caption{The basins of attraction of the fixed points}%
\label{fixed3}%
\end{center}
\end{figure}
\[
\overline{W}(1,0)=\underline{W}(1,0)=\mathbf{B}^{2},
\]
thus the fixed point $(1,0)$ is totally p-attractive and totally n-attractive.
\end{example}

\begin{theorem}
\label{The76}Let $\mu\in\mathbf{B}^{n}$ be a fixed point of $\Phi$. The
following statements are true:

a) We have%
\[
\overline{W}(\mu)=\{\mu^{\prime}|\mu^{\prime}\in\mathbf{B}^{n},\exists
\rho^{\prime}\in P_{n},\underset{t\rightarrow\infty}{\lim}\Phi^{\rho^{\prime}%
}(\mu^{\prime},t)=\mu\},
\]%
\[
\underline{W}(\mu)=\{\mu^{\prime}|\mu^{\prime}\in\mathbf{B}^{n},\forall
\rho^{\prime}\in P_{n},\underset{t\rightarrow\infty}{\lim}\Phi^{\rho^{\prime}%
}(\mu^{\prime},t)=\mu\};
\]

b) $\{\mu\}\subset\underline{W}(\mu)\subset\overline{W}(\mu)$ thus $\mu$ is
p-attractive and n-attractive;

c) $\overline{W}(\mu)$ is p-invariant and $\underline{W}(\mu)$ is n-invariant.
\end{theorem}

\begin{proof}
a) If $\mu^{\prime}\in\overline{W}(\mu),$ then $\rho^{\prime}\in P_{n}$ exists
such that $\omega_{\rho^{\prime}}(\mu^{\prime})=\{\mu\}.$ In this situation
from Theorem \ref{The20} b) we infer that $\underset{t\rightarrow\infty}{\lim
}\Phi^{\rho^{\prime}}(\mu^{\prime},t)=\mu,$ thus $\overline{W}(\mu
)\subset\{\mu^{\prime}|\mu^{\prime}\in\mathbf{B}^{n},\exists\rho^{\prime}\in
P_{n},\underset{t\rightarrow\infty}{\lim}\Phi^{\rho^{\prime}}(\mu^{\prime
},t)=\mu\}.$

Conversely, if $\mu^{\prime}\in\mathbf{B}^{n},\rho^{\prime}\in P_{n}$ exist
such that $\underset{t\rightarrow\infty}{\lim}\Phi^{\rho^{\prime}}(\mu
^{\prime},t)=\mu,$ then $\omega_{\rho^{\prime}}(\mu^{\prime})=\{\mu\}$ from
Definition \ref{Def19} and we get $\{\mu^{\prime}|\mu^{\prime}\in
\mathbf{B}^{n},\exists\rho^{\prime}\in P_{n},\underset{t\rightarrow\infty
}{\lim}\Phi^{\rho^{\prime}}(\mu^{\prime},t)=\mu\}\subset\overline{W}(\mu).$

b) The fact that $\mu\in\underline{W}(\mu)$ is a consequence of the fact that
$\forall\rho^{\prime}\in P_{n},\underset{t\rightarrow\infty}{\lim}\Phi
^{\rho^{\prime}}(\mu,t)=\mu$ (see Corollary \ref{Cor17}).

c) $\mu$ is p-attractive from b), thus $\overline{W}(\mu)$ is p-invariant
(Theorem \ref{The140}).
\end{proof}

\section{The basin of attraction of the orbits and of the $\omega-$limit sets}

\begin{definition}
Let be $\Phi:\mathbf{B}^{n}\rightarrow\mathbf{B}^{n},$ $\mu\in\mathbf{B}^{n}$
and $\rho\in P_{n}.$ We define the \textbf{basins }(or \textbf{kingdoms}, or
\textbf{domains}) \textbf{of p-attraction} of $\Phi^{\rho}(\mu,\cdot
),\omega_{\rho}(\mu)$ by%
\begin{equation}
\overline{W}[\Phi^{\rho}(\mu,\cdot)]=\{\mu^{\prime}|\mu^{\prime}\in
\mathbf{B}^{n},\exists\rho^{\prime}\in P_{n},\exists t^{\prime}\in
\mathbf{R},\forall t\geq t^{\prime}, \label{invo1}%
\end{equation}%
\[
\Phi^{\rho^{\prime}}(\mu^{\prime},t)=\Phi^{\rho}(\mu,t)\},
\]%
\begin{equation}
\overline{W}[\omega_{\rho}(\mu)]=\{\mu^{\prime}|\mu^{\prime}\in\mathbf{B}%
^{n},\exists\rho^{\prime}\in P_{n},\omega_{\rho^{\prime}}(\mu^{\prime}%
)=\omega_{\rho}(\mu)\} \label{invo2}%
\end{equation}
and the \textbf{basins }(or \textbf{kingdoms}, or \textbf{domains}) \textbf{of
n-attraction} of $\Phi^{\rho}(\mu,\cdot),$ $\omega_{\rho}(\mu)$ respectively
by%
\begin{equation}
\underline{W}[\Phi^{\rho}(\mu,\cdot)]=\{\mu^{\prime}|\mu^{\prime}\in
\mathbf{B}^{n},\forall\rho^{\prime}\in P_{n},\exists t^{\prime}\in
\mathbf{R},\forall t\geq t^{\prime}, \label{invo3}%
\end{equation}%
\[
\Phi^{\rho^{\prime}}(\mu^{\prime},t)=\Phi^{\rho}(\mu,t)\},
\]%
\begin{equation}
\underline{W}[\omega_{\rho}(\mu)]=\{\mu^{\prime}|\mu^{\prime}\in\mathbf{B}%
^{n},\forall\rho^{\prime}\in P_{n},\omega_{\rho^{\prime}}(\mu^{\prime}%
)=\omega_{\rho}(\mu)\}. \label{invo4}%
\end{equation}

\end{definition}

\begin{remark}
The attractiveness of the orbits and of the $\omega-$limit sets is defined in
the spirit of Definition \ref{Def168_} and it is their property of making one
of the previous basins of attraction non-empty.

We mention \cite{bib3}, page 133, where $M$ is a differentiable manifold
together with a distance $d$ on $M$ and a discrete time dynamical system is
generated by the $C^{r}-$diffeomorphism $\Phi:M\rightarrow M.$ The orbit
through $x_{0}\in M$ is called attractive if%
\begin{equation}
\exists\delta>0,\forall x\in B(x_{0},\delta),\underset{n\rightarrow\infty
}{\lim}d(\Phi_{n}(x),\Phi_{n}(x_{0}))=0, \label{asterisc3}%
\end{equation}
where $B(x_{0},\delta)$ is the notation for the open ball of center $x_{0}$
and radius $\delta$. In the same work \cite{bib3}, page 133 the orbit through
$x_{0}\in M$ is called stable if
\begin{equation}
\forall\varepsilon>0,\exists\delta(\varepsilon)>0,\forall x\in B(x_{0}%
,\delta),\forall n\in\mathbf{N},d(\Phi_{n}(x),\Phi_{n}(x_{0}))<\varepsilon.
\label{asterisc4}%
\end{equation}
The translation of (\ref{asterisc3}), (\ref{asterisc4}) in our framework gives
the statements%
\begin{equation}
\exists\mu^{\prime}\in\mathbf{B}^{n},\exists\rho^{\prime}\in P_{n}%
,\omega_{\rho^{\prime}}(\mu^{\prime})=\omega_{\rho}(\mu), \label{invo5_}%
\end{equation}%
\begin{equation}
\exists\mu^{\prime}\in\mathbf{B}^{n},\forall\rho^{\prime}\in P_{n}%
,\omega_{\rho^{\prime}}(\mu^{\prime})=\omega_{\rho}(\mu), \label{invo6_}%
\end{equation}%
\begin{equation}
\exists\mu^{\prime}\in\mathbf{B}^{n},\exists\rho^{\prime}\in P_{n},\forall
t\in\mathbf{R},\Phi^{\rho^{\prime}}(\mu^{\prime},t)=\Phi^{\rho}(\mu,t),
\label{invo5}%
\end{equation}%
\begin{equation}
\exists\mu^{\prime}\in\mathbf{B}^{n},\forall\rho^{\prime}\in P_{n},\forall
t\in\mathbf{R},\Phi^{\rho^{\prime}}(\mu^{\prime},t)=\Phi^{\rho}(\mu,t);
\label{invo6}%
\end{equation}
we note that (\ref{invo5_}), (\ref{invo6_}) are equivalent with $\overline
{W}[\omega_{\rho}(\mu)]\neq\emptyset,$ $\underline{W}[\omega_{\rho}(\mu
)]\neq\emptyset$ attractiveness, while (\ref{invo5}), (\ref{invo6}) are
stronger than the attractiveness properties $\overline{W}[\Phi^{\rho}%
(\mu,\cdot)]\neq\emptyset,$ $\underline{W}[\Phi^{\rho}(\mu,\cdot
)]\neq\emptyset$. On the other hand, if we take in (\ref{invo5_}) and
(\ref{invo5}) $\mu^{\prime}=\mu,\rho^{\prime}=\rho$ we get that these two
properties are always true, see Theorem \ref{The146} to follow, items a), b).

Note that the stability of the sets $A$ from the dynamical systems theory is
interpreted as invariance \cite{bib8}, while the stability of the orbits from
the dynamical systems theory is interpreted to be stronger than attractiveness.
\end{remark}

\begin{example}
In Figure \ref{invariance1_}
\begin{figure}
[ptb]
\begin{center}
\fbox{\includegraphics[
height=0.8864in,
width=1.324in
]%
{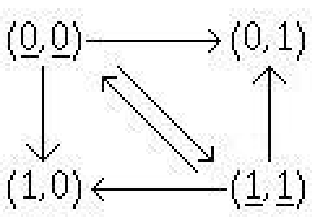}%
}\caption{The basins of attraction of the orbits and of the $\omega$-limit
sets}%
\label{invariance1_}%
\end{center}
\end{figure}
for%
\[
\rho(t)=(1,1)\cdot\chi_{\{0\}}(t)\oplus(1,1)\cdot\chi_{\{1\}}(t)\oplus
(1,1)\cdot\chi_{\{2\}}(t)\oplus...
\]
we get%
\[
\overline{W}[\Phi^{\rho}((0,0),\cdot)]=\overline{W}[\omega_{\rho
}((0,0))]=\{(0,0),(1,1)\},
\]%
\[
\underline{W}[\Phi^{\rho}((0,0),\cdot)]=\underline{W}[\omega_{\rho
}((0,0))]=\emptyset.
\]

\end{example}

\begin{example}
We take in Figure \ref{invariance3}
\begin{figure}
[ptb]
\begin{center}
\fbox{\includegraphics[
height=0.9037in,
width=1.241in
]%
{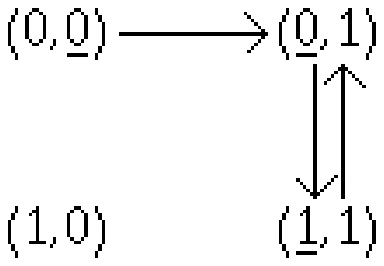}%
}\caption{The basins of attraction of the orbits and of the $\omega$-limit
sets}%
\label{invariance3}%
\end{center}
\end{figure}
\[
\rho(t)=(1,1)\cdot\chi_{\{0\}}(t)\oplus(1,1)\cdot\chi_{\{1\}}(t)\oplus
(1,1)\cdot\chi_{\{2\}}(t)\oplus...
\]
and we obtain%
\[
\overline{W}[\Phi^{\rho}((0,1),\cdot)]=\overline{W}[\omega_{\rho
}((0,1))]=\underline{W}[\omega_{\rho}((0,1))]=\{(0,0),(0,1),(1,1)\},
\]%
\[
\underline{W}[\Phi^{\rho}((0,1),\cdot)]=\emptyset.
\]

\end{example}

\begin{example}
In Figure \ref{invariance2_}
\begin{figure}
[ptb]
\begin{center}
\fbox{\includegraphics[
height=0.8864in,
width=1.324in
]%
{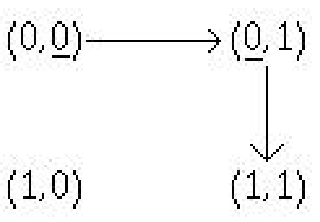}%
}\caption{The basins of attraction of the orbits and of the $\omega$-limit
sets}%
\label{invariance2_}%
\end{center}
\end{figure}
for%
\[
\rho(t)=(0,1)\cdot\chi_{\{0\}}(t)\oplus(1,1)\cdot\chi_{\{1\}}(t)\oplus
(0,1)\cdot\chi_{\{2\}}(t)\oplus(1,1)\cdot\chi_{\{3\}}(t)\oplus...
\]
we see that%
\[
\overline{W}[\Phi^{\rho}((0,0),\cdot)]=\overline{W}[\omega_{\rho
}((0,0))]=\underline{W}[\Phi^{\rho}((0,0),\cdot)]
\]%
\[
=\underline{W}[\omega_{\rho}((0,0))]=\{(0,0),(0,1),(1,1)\}.
\]

\end{example}

\begin{theorem}
\label{The146}We consider the point $\mu\in\mathbf{B}^{n}$ and the function
$\rho\in P_{n}.$

a) $\overline{W}[\Phi^{\rho}(\mu,\cdot)]=\overline{W}[\omega_{\rho}(\mu)];$

b) we have $Or_{\rho}(\mu)\subset\overline{W}[\Phi^{\rho}(\mu,\cdot)],$ thus
$\overline{W}[\Phi^{\rho}(\mu,\cdot)]$ is non-empty;

c) $\underline{W}[\Phi^{\rho}(\mu,\cdot)]\subset\underline{W}[\omega_{\rho
}(\mu)]$;

d) $\underline{W}[\Phi^{\rho}(\mu,\cdot)]\neq\emptyset\Longleftrightarrow
card(\omega_{\rho}(\mu))=1,$ thus $card(\omega_{\rho}(\mu))=1$ implies that
$\underline{W}[\Phi^{\rho}(\mu,\cdot)]$ and $\underline{W}[\omega_{\rho}%
(\mu)]$ are non-empty;

e) if $\exists\mu^{\prime}\in\mathbf{B}^{n},\omega_{\rho}(\mu)=\{\mu^{\prime
}\},$ then $\underline{W}[\Phi^{\rho}(\mu,\cdot)]=\underline{W}[\omega_{\rho
}(\mu)]=\underline{W}(\mu^{\prime}).$
\end{theorem}

\begin{proof}
a) Let $\mu^{\prime}\in\overline{W}[\Phi^{\rho}(\mu,\cdot)]$ be arbitrary, for
which $\rho^{\prime}\in P_{n},t^{\prime}\in\mathbf{R}$ exist such that%
\begin{equation}
\forall t\geq t^{\prime},\Phi^{\rho^{\prime}}(\mu^{\prime},t)=\Phi^{\rho}%
(\mu,t), \label{*_}%
\end{equation}%
\begin{equation}
\omega_{\rho^{\prime}}(\mu^{\prime})=\{\Phi^{\rho^{\prime}}(\mu^{\prime
},t)|t\geq t^{\prime}\}, \label{**_}%
\end{equation}%
\begin{equation}
\omega_{\rho}(\mu)=\{\Phi^{\rho}(\mu,t)|t\geq t^{\prime}\}. \label{***_}%
\end{equation}
(\ref{*_}) is fulfilled from the definition (\ref{invo1}) of $\overline
{W}[\Phi^{\rho}(\mu,\cdot)]$ and, by taking $t^{\prime}$ sufficiently great,
(\ref{**_}), (\ref{***_}) are fulfilled too (Theorem \ref{The20_}). As
$\omega_{\rho^{\prime}}(\mu^{\prime})=\omega_{\rho}(\mu)$ we get $\mu^{\prime
}\in\overline{W}[\omega_{\rho}(\mu)]$ and because $\mu^{\prime}$ was
arbitrary, we infer that $\overline{W}[\Phi^{\rho}(\mu,\cdot)]\subset
\overline{W}[\omega_{\rho}(\mu)]$.

Conversely, let $\mu^{\prime}\in\overline{W}[\omega_{\rho}(\mu)]$ be
arbitrary, thus%
\[
\exists\rho^{\prime}\in P_{n},\omega_{\rho^{\prime}}(\mu^{\prime}%
)=\omega_{\rho}(\mu)
\]
and let $\mu^{\prime\prime}\in\omega_{\rho^{\prime}}(\mu^{\prime}%
)=\omega_{\rho}(\mu)$ be some point,%
\begin{equation}
\Phi^{\rho^{\prime}}(\mu^{\prime},t_{1})=\Phi^{\rho}(\mu,t_{2})=\mu
^{\prime\prime}, \label{bas40}%
\end{equation}
$t_{1},t_{2}\in\mathbf{R}$. The function%
\begin{equation}
\rho^{\prime\prime}(t)=\rho^{\prime}(t-t_{2}+t_{1})\cdot\chi_{(-\infty,t_{2}%
]}(t)\oplus\rho(t)\cdot\chi_{(t_{2},\infty)}(t) \label{bas41}%
\end{equation}
is progressive and fulfills%
\[
\Phi^{\rho^{\prime\prime}}(\mu^{\prime},t_{2})\overset{(\ref{bas41})}{=}%
\Phi^{\rho^{\prime}\circ\tau^{t_{2}-t_{1}}}(\mu^{\prime},t_{2})\overset
{Theorem\ \ref{The11}}{=}\Phi^{\rho^{\prime}}(\mu^{\prime},t_{1}%
)\overset{(\ref{bas40})}{=}\Phi^{\rho}(\mu,t_{2})\overset{(\ref{bas40})}{=}%
\mu^{\prime\prime},
\]%
\[
\forall t>t_{2},\Phi^{\rho^{\prime\prime}}(\mu^{\prime},t)\overset
{Theorem\ \ref{The10}\ d)}{=}\Phi^{\rho^{\prime\prime}\cdot\chi_{(t_{2}%
,\infty)}}(\mu^{\prime\prime},t)
\]%
\[
\overset{(\ref{bas41})}{=}\Phi^{\rho\cdot\chi_{(t_{2},\infty)}}(\mu
^{\prime\prime},t)\overset{Theorem\ \ref{The10}\ d)}{=}\Phi^{\rho}(\mu,t)
\]
in other words $\mu^{\prime}\in\overline{W}[\Phi^{\rho}(\mu,\cdot)].$ The fact
that $\mu^{\prime}$ was arbitrary gives the conclusion that $\overline
{W}[\omega_{\rho}(\mu)]\subset\overline{W}[\Phi^{\rho}(\mu,\cdot)].$

b) Let $\mu^{\prime}\in Or_{\rho}(\mu)$ be arbitrary, thus $\exists t^{\prime
}\in\mathbf{R}$ with $\mu^{\prime}=\Phi^{\rho}(\mu,t^{\prime}).$ We get%
\[
\forall t\geq t^{\prime},\Phi^{\rho\cdot\chi_{(t^{\prime},\infty)}}%
(\mu^{\prime},t)\overset{Theorem\ \ref{The10}\ d)}{=}\Phi^{\rho}(\mu,t).
\]
We have shown that $\mu^{\prime}\in\overline{W}[\Phi^{\rho}(\mu,\cdot)]$ and
as $\mu^{\prime}$ was arbitrarily chosen, we infer $Or_{\rho}(\mu
)\subset\overline{W}[\Phi^{\rho}(\mu,\cdot)].$

c) Let $\mu^{\prime}\in\underline{W}[\Phi^{\rho}(\mu,\cdot)]$ and
$\rho^{\prime}\in P_{n}$ be arbitrary. Some sufficiently great $t^{\prime}%
\in\mathbf{R}$ exists such that%
\[
\forall t\geq t^{\prime},\Phi^{\rho^{\prime}}(\mu^{\prime},t)=\Phi^{\rho}%
(\mu,t)
\]
and we have%
\[
\omega_{\rho^{\prime}}(\mu^{\prime})=\{\Phi^{\rho^{\prime}}(\mu^{\prime
},t)|t\geq t^{\prime}\}=\{\Phi^{\rho}(\mu,t)|t\geq t^{\prime}\}=\omega_{\rho
}(\mu),
\]
i.e. $\mu^{\prime}\in\underline{W}[\omega_{\rho}(\mu)].$

d) $\Longrightarrow$ Let be $\rho$ given by%
\[
\rho(t)=\alpha^{0}\cdot\chi_{\{t_{0}\}}(t)\oplus...\oplus\alpha^{k}\cdot
\chi_{\{t_{k}\}}(t)\oplus...,
\]
$\alpha\in\Pi_{n},(t_{k})\in Seq$ and we define%
\[
\rho^{\prime}(t)=\alpha^{0}\cdot\chi_{\{t_{0}^{\prime}\}}(t)\oplus
...\oplus\alpha^{k}\cdot\chi_{\{t_{k}^{\prime}\}}(t)\oplus...,
\]
where%
\begin{equation}
t_{k}^{\prime}=\frac{t_{k}+t_{k+1}}{2},k\in\mathbf{N} \label{bas1}%
\end{equation}
belongs to $Seq$. We call point of discontinuity of $\Phi^{\rho}(\mu,\cdot)$ a
point $\xi\in\mathbf{R}$ with the property that $\mu^{\prime},\mu
^{\prime\prime}\in\mathbf{B}^{n}$ and $\varepsilon>0$ exist such that%
\[
\forall t\in(\xi-\varepsilon,\xi),\Phi^{\rho}(\mu,t)=\mu^{\prime},
\]%
\[
\forall t\in\lbrack\xi,\xi+\varepsilon),\Phi^{\rho}(\mu,t)=\mu^{\prime\prime
},
\]%
\[
\mu^{\prime}\neq\mu^{\prime\prime}.
\]
The hypothesis states that $\widetilde{\mu}\in\underline{W}[\Phi^{\rho}%
(\mu,\cdot)]$ exists fulfilling the property%
\begin{equation}
\exists t^{\prime}\in\mathbf{R},\forall t\geq t^{\prime},\Phi^{\rho^{\prime}%
}(\widetilde{\mu},t)=\Phi^{\rho}(\mu,t). \label{bas1_}%
\end{equation}
Let us suppose against all reason that $card(\omega_{\rho^{\prime}}%
(\widetilde{\mu}))=card(\omega_{\rho}(\mu))>1$ and%
\[
\omega_{\rho^{\prime}}(\widetilde{\mu})=\{\Phi^{\rho^{\prime}}(\widetilde{\mu
},t)|t\geq t^{\prime\prime}\}=\{\Phi^{\rho}(\mu,t)|t\geq t^{\prime\prime
}\}=\omega_{\rho}(\mu),
\]
$t^{\prime\prime}\geq t^{\prime}.$ Then equation (\ref{bas1_}) is
contradictory, since the discontinuity points of $\Phi^{\rho^{\prime}%
}(\widetilde{\mu},\cdot)_{|[t^{\prime},\infty)}$\footnote{$\Phi^{\rho^{\prime
}}(\widetilde{\mu},\cdot)_{|[t^{\prime},\infty)}$ is the notation for the
restriction of $\Phi^{\rho^{\prime}}(\widetilde{\mu},\cdot):\mathbf{R}%
\rightarrow\mathbf{B}^{n}$ to $[t^{\prime},\infty).$
\par
{}} and $\Phi^{\rho}(\mu,\cdot)_{|[t^{\prime},\infty)}$ are included in the
disjoint sets $[t^{\prime},\infty)\cap(t_{k}^{\prime})$ and $[t^{\prime
},\infty)\cap(t_{k}).$ The conclusion is that $card(\omega_{\rho^{\prime}%
}(\widetilde{\mu}))=card(\omega_{\rho}(\mu))=1$ and in (\ref{bas1_}) the
disjoint sets $[t^{\prime},\infty)\cap(t_{k}^{\prime})$ and $[t^{\prime
},\infty)\cap(t_{k})$ contain no discontinuity points.

$\Longleftarrow$ We presume that $\mu^{\prime}\in\mathbf{B}^{n}$ exists with
$\omega_{\rho}(\mu)=\{\mu^{\prime}\}$ and then for an arbitrary $\rho^{\prime
}\in P_{n}$ we get $\forall t\in\mathbf{R},\Phi^{\rho^{\prime}}(\mu^{\prime
},t)=\mu^{\prime}.$ As $\underset{t\rightarrow\infty}{\lim}\Phi^{\rho}%
(\mu,t)=\mu^{\prime},$ we conclude that $\exists t^{\prime}\in\mathbf{R}$ such
that $\forall t\geq t^{\prime},$%
\[
\Phi^{\rho}(\mu,t)=\Phi^{\rho^{\prime}}(\mu^{\prime},t)=\mu^{\prime},
\]
thus $\mu^{\prime}\in\underline{W}[\Phi^{\rho}(\mu,\cdot)].$

e) The fact that $\omega_{\rho}(\mu)=\{\mu^{\prime}\}$ shows that $\mu
^{\prime}$ is a fixed point of $\Phi$ (Theorem \ref{The20} b)), thus
$\underline{W}(\mu^{\prime})\neq\emptyset.$ $\underline{W}[\Phi^{\rho}%
(\mu,\cdot)],$ $\underline{W}[\omega_{\rho}(\mu)]$ and $\underline{W}%
(\mu^{\prime})$ are all equal with the set%
\[
\{\mu^{\prime\prime}|\mu^{\prime\prime}\in\mathbf{B}^{n},\forall\rho^{\prime
}\in P_{n},\underset{t\rightarrow\infty}{\lim}\Phi^{\rho^{\prime}}(\mu
^{\prime\prime},t)=\mu^{\prime}\}.
\]

\end{proof}

\begin{theorem}
For any $\mu\in\mathbf{B}^{n}$ and any $\rho\in P_{n},$

a) the basin of p-attraction $\overline{W}[\Phi^{\rho}(\mu,\cdot)]$ is p-invariant;

b) if $\Phi^{\rho}(\mu,\cdot)$ is n-attractive, then the basin of n-attraction
$\underline{W}[\Phi^{\rho}(\mu,\cdot)]$ is n-invariant;

c) if $\omega_{\rho}(\mu)$ is n-attractive, then $\underline{W}[\omega_{\rho
}(\mu)]$ is n-invariant.
\end{theorem}

\begin{proof}
a) From $Or_{\rho}(\mu)\neq\emptyset$ and $Or_{\rho}(\mu)\subset\overline
{W}[\Phi^{\rho}(\mu,\cdot)],$ see Theorem \ref{The146} b), we have that
$\overline{W}[\Phi^{\rho}(\mu,\cdot)]\neq\emptyset.$ Let $\mu^{\prime}%
\in\overline{W}[\Phi^{\rho}(\mu,\cdot)]$ be arbitrary$,$ meaning that%
\[
\exists\rho^{\prime}\in P_{n},\exists t^{\prime}\in\mathbf{R},\forall t\geq
t^{\prime},\Phi^{\rho^{\prime}}(\mu^{\prime},t)=\Phi^{\rho}(\mu,t)
\]
and we prove the inclusion $Or_{\rho^{\prime}}(\mu^{\prime})\subset
\overline{W}[\Phi^{\rho}(\mu,\cdot)].$ Indeed, the points $\mu^{\prime\prime}$
of the orbit $Or_{\rho^{\prime}}(\mu^{\prime})$ are of the form
\[
\exists t_{1}\in\mathbf{R},\mu^{\prime\prime}=\Phi^{\rho^{\prime}}(\mu
^{\prime},t_{1})
\]
and they fulfill%
\[
\forall t\geq t_{1},\Phi^{\rho^{\prime}}(\mu^{\prime},t)\overset
{Theorem\ \ref{The10}\ d)}{=}\Phi^{\rho^{\prime}\cdot\chi_{(t_{1},\infty)}%
}(\mu^{\prime\prime},t),
\]
thus%
\[
\forall t\geq\max\{t^{\prime},t_{1}\},\Phi^{\rho}(\mu,t)=\Phi^{\rho^{\prime}%
}(\mu^{\prime},t)=\Phi^{\rho^{\prime}\cdot\chi_{(t_{1},\infty)}}(\mu
^{\prime\prime},t).
\]
We infer that $\mu^{\prime\prime}\in\overline{W}[\Phi^{\rho}(\mu,\cdot)].$

b) We must prove any of the following three equivalent statements:%
\begin{equation}
\forall\mu^{\prime}\in\underline{W}[\Phi^{\rho}(\mu,\cdot)],\forall
\rho^{\prime}\in P_{n},Or_{\rho^{\prime}}(\mu^{\prime})\subset\underline
{W}[\Phi^{\rho}(\mu,\cdot)],
\end{equation}%
\begin{equation}
\forall\mu^{\prime}\in\underline{W}[\Phi^{\rho}(\mu,\cdot)],\forall
\rho^{\prime}\in P_{n},\forall t_{1}\in\mathbf{R},\Phi^{\rho^{\prime}}%
(\mu^{\prime},t_{1})\in\underline{W}[\Phi^{\rho}(\mu,\cdot)],
\end{equation}%
\begin{equation}
\forall\mu^{\prime}\in\mathbf{B}^{n},\forall\rho^{\prime\prime}\in
P_{n},\exists t^{\prime}\in\mathbf{R},\forall t\geq t^{\prime},\Phi
^{\rho^{\prime\prime}}(\mu^{\prime},t)=\Phi^{\rho}(\mu,t)\Longrightarrow
\label{bas6}%
\end{equation}%
\[
\Longrightarrow\forall\rho^{\prime}\in P_{n},\forall t_{1}\in\mathbf{R}%
,\forall\rho^{\prime\prime\prime}\in P_{n},\exists t^{\prime\prime}%
\in\mathbf{R},
\]%
\[
\forall t\geq t^{\prime\prime},\Phi^{\rho^{\prime\prime\prime}}(\Phi
^{\rho^{\prime}}(\mu^{\prime},t_{1}),t)=\Phi^{\rho}(\mu,t).
\]
For this, let $\mu^{\prime}\in\mathbf{B}^{n}$ be arbitrary, fixed, making the
following property true:%
\begin{equation}
\forall\rho^{\prime\prime}\in P_{n},\exists t^{\prime}\in\mathbf{R},\forall
t\geq t^{\prime},\Phi^{\rho^{\prime\prime}}(\mu^{\prime},t)=\Phi^{\rho}(\mu,t)
\label{bas7}%
\end{equation}
and we take $\rho^{\prime}\in P_{n},t_{1}\in\mathbf{R},\rho^{\prime
\prime\prime}\in P_{n}$ arbitrarily. The truth of (\ref{bas7}) for%
\begin{equation}
\rho^{\prime\prime}=\rho^{\prime}\cdot\chi_{(-\infty,t_{1}]}\oplus\rho
^{\prime\prime\prime}\cdot\chi_{(t_{1},\infty)} \label{bas2}%
\end{equation}
shows the existence of $t^{\prime}\in\mathbf{R}$ that we choose $>t_{1}$ such
that $\forall t\geq t^{\prime},$%
\[
\Phi^{\rho}(\mu,t)\overset{(\ref{bas7})}{=}\Phi^{\rho^{\prime\prime}}%
(\mu^{\prime},t)\overset{Theorem\ \ref{The10}\ d)}{=}\Phi^{\rho^{\prime\prime
}\cdot\chi_{(t_{1},\infty)}}(\Phi^{\rho^{\prime\prime}}(\mu^{\prime}%
,t_{1}),t)=
\]%
\[
\overset{(\ref{bas2})}{=}\Phi^{\rho^{\prime\prime\prime}\cdot\chi
_{(t_{1},\infty)}}(\Phi^{\rho^{\prime\prime}}(\mu^{\prime},t_{1}%
),t)\overset{(\ref{bas2})}{=}\Phi^{\rho^{\prime\prime\prime}\cdot\chi
_{(t_{1},\infty)}}(\Phi^{\rho^{\prime}}(\mu^{\prime},t_{1}),t)
\]%
\[
\overset{Theorem\ \ref{The10}\ d)}{=}\Phi^{\rho^{\prime\prime\prime}}%
(\Phi^{\rho^{\prime}}(\mu^{\prime},t_{1}),t).
\]
(\ref{bas6}) holds.

c) We must prove any of the equivalent statements:%
\begin{equation}
\forall\mu^{\prime}\in\underline{W}[\omega_{\rho}(\mu)],\forall\rho^{\prime
}\in P_{n},Or_{\rho^{\prime}}(\mu^{\prime})\subset\underline{W}[\omega_{\rho
}(\mu)],
\end{equation}%
\begin{equation}
\forall\mu^{\prime}\in\underline{W}[\omega_{\rho}(\mu)],\forall\rho^{\prime
}\in P_{n},\forall t_{1}\in\mathbf{R},\Phi^{\rho^{\prime}}(\mu^{\prime}%
,t_{1})\in\underline{W}[\omega_{\rho}(\mu)],
\end{equation}%
\begin{equation}
\forall\mu^{\prime}\in\mathbf{B}^{n},\forall\rho^{\prime\prime}\in
P_{n},\omega_{\rho^{\prime\prime}}(\mu^{\prime})=\omega_{\rho}(\mu
)\Longrightarrow\label{bas8}%
\end{equation}%
\[
\Longrightarrow\forall\rho^{\prime}\in P_{n},\forall t_{1}\in\mathbf{R}%
,\forall\rho^{\prime\prime\prime}\in P_{n},\omega_{\rho^{\prime\prime\prime}%
}(\Phi^{\rho^{\prime}}(\mu^{\prime},t_{1}))=\omega_{\rho}(\mu).
\]
Let $\mu^{\prime}\in\mathbf{B}^{n}$ be arbitrary, fixed, that fulfills%
\begin{equation}
\forall\rho^{\prime\prime}\in P_{n},\omega_{\rho^{\prime\prime}}(\mu^{\prime
})=\omega_{\rho}(\mu) \label{bas9}%
\end{equation}
and we take some arbitrary $\rho^{\prime}\in P_{n},t_{1}\in\mathbf{R}%
,\rho^{\prime\prime\prime}\in P_{n}$. A number $t_{2}$ exists with the
property $\forall t\leq t_{2},\rho^{\prime\prime\prime}(t)=0.$ We define
$\widetilde{\rho}\in P_{n}$ by%
\[
\widetilde{\rho}(t)=\rho^{\prime\prime\prime}(t-t_{1}+t_{2})
\]
and we note that (as far as $\forall t\leq t_{1},t-t_{1}+t_{2}\leq t_{2}$) we
have $\widetilde{\rho}\cdot\chi_{(t_{1},\infty)}(t)=\rho^{\prime\prime\prime
}(t-t_{1}+t_{2})=(\rho^{\prime\prime\prime}\circ\tau^{t_{1}-t_{2}})(t),$ thus
\begin{equation}
\omega_{\widetilde{\rho}\cdot\chi_{(t_{1},\infty)}}(\Phi^{\rho^{\prime}}%
(\mu^{\prime},t_{1}))=\omega_{\rho^{\prime\prime\prime}\circ\tau^{t_{1}-t_{2}%
}}(\Phi^{\rho^{\prime}}(\mu^{\prime},t_{1}))=\omega_{\rho^{\prime\prime\prime
}}(\Phi^{\rho^{\prime}}(\mu^{\prime},t_{1})), \label{bas3}%
\end{equation}
see Theorem \ref{The22}. The truth of (\ref{bas9}) for%
\begin{equation}
\rho^{\prime\prime}=\rho^{\prime}\cdot\chi_{(-\infty,t_{1}]}\oplus
\widetilde{\rho}\cdot\chi_{(t_{1},\infty)} \label{bas4}%
\end{equation}
shows that%
\[
\omega_{\rho}(\mu)\overset{(\ref{bas9})}{=}\omega_{\rho^{\prime\prime}}%
(\mu^{\prime})\overset{(\ref{bas4})}{=}\omega_{\rho^{\prime}\cdot
\chi_{(-\infty,t_{1}]}\oplus\widetilde{\rho}\cdot\chi_{(t_{1},\infty)}}%
(\mu^{\prime})=
\]%
\[
\overset{(\ref{asterisc})}{=}\omega_{\widetilde{\rho}\cdot\chi_{(t_{1}%
,\infty)}}(\Phi^{\rho^{\prime}}(\mu^{\prime},t_{1}))\overset{(\ref{bas3})}%
{=}\omega_{\rho^{\prime\prime\prime}}(\Phi^{\rho^{\prime}}(\mu^{\prime}%
,t_{1})).
\]
The statement (\ref{bas8}) was proved.
\end{proof}

\begin{theorem}
Let $\mu\in\mathbf{B}^{n}$ be a fixed point of $\Phi$ and $\rho\in P_{n}$. We
have%
\[
\overline{W}(\mu)=\overline{W}[\Phi^{\rho}(\mu,\cdot)]=\overline{W}%
[\omega_{\rho}(\mu)],
\]%
\[
\underline{W}(\mu)=\underline{W}[\Phi^{\rho}(\mu,\cdot)]=\underline{W}%
[\omega_{\rho}(\mu)].
\]

\end{theorem}

\begin{proof}
Because%
\[
\forall t\in\mathbf{R},\Phi^{\rho}(\mu,t)=\mu,
\]%
\[
\omega_{\rho}(\mu)=\{\mu\}
\]
we get for any $\mu^{\prime}\in\mathbf{B}^{n}$ the equivalence of the
statements%
\[
\exists\rho^{\prime}\in P_{n},\omega_{\rho^{\prime}}(\mu^{\prime})\subset
\{\mu\}\text{ (equation (\ref{sta9}))},
\]%
\[
\exists\rho^{\prime}\in P_{n},\exists t^{\prime}\in\mathbf{R},\forall t\geq
t^{\prime},\Phi^{\rho^{\prime}}(\mu^{\prime},t)=\mu\text{ (see (\ref{invo1}%
))},
\]%
\[
\exists\rho^{\prime}\in P_{n},\omega_{\rho^{\prime}}(\mu^{\prime}%
)=\{\mu\}\text{ (see (\ref{invo2}))}%
\]
meaning that $\mu^{\prime}\in\overline{W}(\mu),\mu^{\prime}\in\overline
{W}[\Phi^{\rho}(\mu,\cdot)],\mu^{\prime}\in\overline{W}[\omega_{\rho}(\mu)]$,
thus $\overline{W}(\mu)=\overline{W}[\Phi^{\rho}(\mu,\cdot)]=\overline
{W}[\omega_{\rho}(\mu)].$
\end{proof}

\begin{theorem}
\label{The105}Let be $\mu\in\mathbf{B}^{n}$ and $\rho\in P_{n}$. The following
statements hold:

a) $\overline{W}[\Phi^{\rho}(\mu,\cdot)]=\overline{W}(Or_{\rho}(\mu));$

b) $\underline{W}[\Phi^{\rho}(\mu,\cdot)]\subset\underline{W}(Or_{\rho}%
(\mu));$

c) $\overline{W}[\omega_{\rho}(\mu)]=\overline{W}(\omega_{\rho}(\mu));$

d) $\underline{W}[\omega_{\rho}(\mu)]\subset\underline{W}(\omega_{\rho}%
(\mu)).$
\end{theorem}

\begin{proof}
a) We prove that $\overline{W}[\Phi^{\rho}(\mu,\cdot)]\subset\overline
{W}(Or_{\rho}(\mu))$ and let $\mu^{\prime}\in\overline{W}[\Phi^{\rho}%
(\mu,\cdot)]$ be arbitrary, thus $\mu^{\prime}\in\overline{W}[\omega_{\rho
}(\mu)]$ (Theorem \ref{The146} a)). We get $\exists\rho^{\prime}\in
P_{n},\omega_{\rho^{\prime}}(\mu^{\prime})=\omega_{\rho}(\mu)\subset Or_{\rho
}(\mu)$ and finally $\mu^{\prime}\in\overline{W}(Or_{\rho}(\mu)).$

We prove now that $\overline{W}(Or_{\rho}(\mu))\subset\overline{W}[\Phi^{\rho
}(\mu,\cdot)].$ We presume that $\mu^{\prime}\in\overline{W}(Or_{\rho}(\mu)),$
i.e. $\exists\rho^{\prime}\in P_{n},\omega_{\rho^{\prime}}(\mu^{\prime
})\subset Or_{\rho}(\mu).$ Let $\mu^{\prime\prime}\in\omega_{\rho^{\prime}%
}(\mu^{\prime})$ be arbitrary. $t_{1}\in\mathbf{R},t_{2}\in\mathbf{R}$ exist
then such that%
\begin{equation}
\Phi^{\rho^{\prime}}(\mu^{\prime},t_{1})=\Phi^{\rho}(\mu,t_{2})=\mu
^{\prime\prime} \label{bas21}%
\end{equation}
and we define%
\begin{equation}
\rho^{\prime\prime}(t)=\rho^{\prime}(t-t_{2}+t_{1})\cdot\chi_{(-\infty,t_{2}%
]}(t)\oplus\rho(t)\cdot\chi_{(t_{2},\infty)}(t). \label{bas22}%
\end{equation}
We note that $\rho^{\prime\prime}\in P_{n}.$ We have%
\begin{equation}
\forall t\leq t_{2},\Phi^{\rho^{\prime\prime}}(\mu^{\prime},t)\overset
{(\ref{bas22})}{=}\Phi^{\rho^{\prime}\circ\tau^{t_{2}-t_{1}}}(\mu^{\prime
},t)\overset{Theorem\ \ref{The11}}{=}\Phi^{\rho^{\prime}}(\mu^{\prime}%
,t-t_{2}+t_{1}), \label{bas23}%
\end{equation}%
\begin{equation}
\Phi^{\rho^{\prime\prime}}(\mu^{\prime},t_{2})\overset{(\ref{bas23})}{=}%
\Phi^{\rho^{\prime}}(\mu^{\prime},t_{1})\overset{(\ref{bas21})}{=}\mu
^{\prime\prime}, \label{bas24}%
\end{equation}
thus $\forall t>t_{2},$%
\[
\Phi^{\rho^{\prime\prime}}(\mu^{\prime},t)\overset{(\ref{bas24})}{=}\Phi
^{\rho^{\prime\prime}\cdot\chi_{(t_{2},\infty)}}(\mu^{\prime\prime}%
,t)\overset{(\ref{bas22})}{=}\Phi^{\rho\cdot\chi_{(t_{2},\infty)}}(\mu
^{\prime\prime},t)=\Phi^{\rho}(\mu,t).
\]
We have proved the fact that $\mu^{\prime}\in\overline{W}[\Phi^{\rho}%
(\mu,\cdot)],$ thus $\overline{W}(Or_{\rho}(\mu))\subset\overline{W}%
[\Phi^{\rho}(\mu,\cdot)].$

b) Let $\mu^{\prime}\in\underline{W}[\Phi^{\rho}(\mu,\cdot)]$ be arbitrary, in
other words $\exists\mu^{\prime\prime}\in\mathbf{B}^{n}$ such that
$\omega_{\rho}(\mu)=\{\mu^{\prime\prime}\}$ (Theorem \ref{The146} d)). We
infer that $\mu^{\prime\prime}\in Or_{\rho}(\mu)$ and%
\[
\underline{W}[\Phi^{\rho}(\mu,\cdot)]\overset{Theorem\text{ }\ref{The146}%
\text{ }e)}{=}\underline{W}(\mu^{\prime\prime})\overset{Theorem\ \ref{The25}%
\ ii)}{\subset}\underline{W}(Or_{\rho}(\mu)).
\]

c) For any $\mu^{\prime}\in\overline{W}[\omega_{\rho}(\mu)]$ we have
$\exists\rho^{\prime}\in P_{n},\omega_{\rho^{\prime}}(\mu^{\prime}%
)=\omega_{\rho}(\mu),$ thus $\exists\rho^{\prime}\in P_{n},\omega
_{\rho^{\prime}}(\mu^{\prime})\subset\omega_{\rho}(\mu)$ proving that
$\mu^{\prime}\in\overline{W}(\omega_{\rho}(\mu))$ and the conclusion is that%
\begin{equation}
\overline{W}[\omega_{\rho}(\mu)]\subset\overline{W}(\omega_{\rho}(\mu)).
\label{bas10}%
\end{equation}
In order to show that the inclusion (\ref{bas10}) takes place under the form
of an equality, we presume against all reason that $\mu^{\prime}\in
\overline{W}(\omega_{\rho}(\mu))\setminus\overline{W}[\omega_{\rho}(\mu)]$
exists, wherefrom we get%
\begin{equation}
\exists\rho^{\prime}\in P_{n},\omega_{\rho^{\prime}}(\mu^{\prime}%
)\subset\omega_{\rho}(\mu), \label{bas11}%
\end{equation}%
\begin{equation}
\forall\rho^{\prime\prime}\in P_{n},\omega_{\rho^{\prime\prime}}(\mu^{\prime
})\neq\omega_{\rho}(\mu). \label{bas12}%
\end{equation}
From (\ref{bas11}), $t_{1}\in\mathbf{R},t_{2}\in\mathbf{R},\mu^{\prime\prime
}\in\mathbf{B}^{n}$ exist such that%
\begin{equation}
\Phi^{\rho^{\prime}}(\mu^{\prime},t_{1})=\Phi^{\rho}(\mu,t_{2})=\mu
^{\prime\prime} \label{bas13}%
\end{equation}
and we define%
\begin{equation}
\rho^{\prime\prime}(t)=\rho^{\prime}(t-t_{2}+t_{1})\cdot\chi_{(-\infty,t_{2}%
]}(t)\oplus\rho(t)\cdot\chi_{(t_{2},\infty)}(t). \label{bas14}%
\end{equation}
Obviously $\rho^{\prime\prime}\in P_{n}.$ We have%
\begin{equation}
\forall t\leq t_{2},\Phi^{\rho^{\prime\prime}}(\mu^{\prime},t)\overset
{(\ref{bas14})}{=}\Phi^{\rho^{\prime}\circ\tau^{t_{2}-t_{1}}}(\mu^{\prime
},t)\overset{Theorem\ \ref{The11}}{=}\Phi^{\rho^{\prime}}(\mu^{\prime}%
,t-t_{2}+t_{1}), \label{bas15}%
\end{equation}%
\begin{equation}
\Phi^{\rho^{\prime\prime}}(\mu^{\prime},t_{2})\overset{(\ref{bas15})}{=}%
\Phi^{\rho^{\prime}}(\mu^{\prime},t_{1})\overset{(\ref{bas13})}{=}\mu
^{\prime\prime}, \label{bas16}%
\end{equation}
thus $\forall t>t_{2},$%
\[
\Phi^{\rho^{\prime\prime}}(\mu^{\prime},t)\overset{(\ref{bas16})}{=}\Phi
^{\rho^{\prime\prime}\cdot\chi_{(t_{2},\infty)}}(\mu^{\prime\prime}%
,t)\overset{(\ref{bas14})}{=}\Phi^{\rho\cdot\chi_{(t_{2},\infty)}}(\mu
^{\prime\prime},t)\overset{Theorem\ \ref{The10}\ d)}{=}\Phi^{\rho}(\mu,t).
\]
We have obtained $\omega_{\rho^{\prime\prime}}(\mu^{\prime})=\omega_{\rho}%
(\mu),$ contradiction with (\ref{bas12}).

d) We take an arbitrary $\mu^{\prime}\in\underline{W}[\omega_{\rho}(\mu)].$
The truth of%
\[
\forall\rho^{\prime}\in P_{n},\omega_{\rho^{\prime}}(\mu^{\prime}%
)=\omega_{\rho}(\mu)
\]
implies that%
\[
\forall\rho^{\prime}\in P_{n},\omega_{\rho^{\prime}}(\mu^{\prime}%
)\subset\omega_{\rho}(\mu)
\]
is true thus $\mu^{\prime}\in\underline{W}(\omega_{\rho}(\mu)).$
\end{proof}

\begin{example}
In Figure \ref{minimality2}
\begin{figure}
[ptb]
\begin{center}
\fbox{\includegraphics[
height=0.9115in,
width=1.3327in
]%
{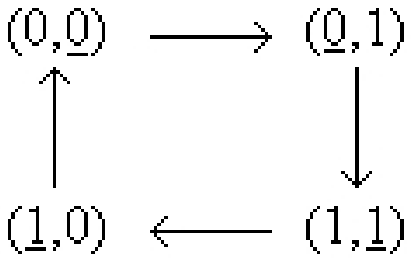}%
}\caption{Showing that the inclusion $\underline{W}[\Phi^{\rho}(\mu
,\cdot)]\subset\underline{W}(Or_{\rho}(\mu))$ is not equality}%
\label{minimality2}%
\end{center}
\end{figure}
for any $\rho\in P_{2}$ we have that $\underline{W}[\Phi^{\rho}((0,0),\cdot
)]=\emptyset,$ $Or_{\rho}(0,0)=\mathbf{B}^{2},$ $\underline{W}(\mathbf{B}%
^{2})=\mathbf{B}^{2}$ and the inclusion from Theorem \ref{The105} b) is not equality.
\end{example}

\begin{example}
In Figure \ref{transitivity3}
\begin{figure}
[ptb]
\begin{center}
\fbox{\includegraphics[
height=0.921in,
width=1.3534in
]%
{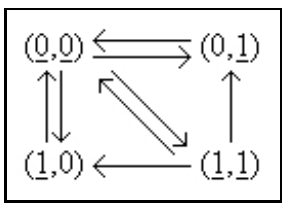}%
}\caption{Showing that the inclusion $\underline{W}[\omega_{\rho}(\mu
)]\subset\underline{W}(\omega_{\rho}(\mu))$ is not equality}%
\label{transitivity3}%
\end{center}
\end{figure}
we take $\mu=(0,0),\omega_{\rho}(\mu)=\mathbf{B}^{2}$ so that $\underline
{W}(\omega_{\rho}(\mu))=\underline{W}(\mathbf{B}^{2})=\mathbf{B}^{2}.$ On the
other hand we can see that $\underline{W}[\omega_{\rho}(\mu)]=\emptyset,$
showing that the inclusion from Theorem \ref{The105} d) is not equality.
\end{example}

\end{document}